%% file: fair_public_decisions.tex
\documentclass{article}
\usepackage{booktabs} 
\usepackage{xspace}
\usepackage{multirow}
\usepackage{amsmath}
\usepackage{amsthm,amssymb}
\usepackage[numbers]{natbib}
\usepackage[small]{caption}

\usepackage{geometry}
\usepackage[ruled,linesnumbered]{algorithm2e}

\SetArgSty{textrm}  
\SetAlFnt{\small}
\SetAlCapFnt{\small}
\SetAlCapNameFnt{\small}
\SetAlCapHSkip{0pt}
\IncMargin{-\parindent}
\usepackage{algpseudocode}

\newtheorem{theorem}{Theorem}

\newtheorem{corollary}[theorem]{Corollary}
\newtheorem{lemma}[theorem]{Lemma}

\newtheorem{example}{Example}

\renewcommand{\epsilon}{\varepsilon}
\newcommand{\eps}{\epsilon}
\mathchardef\mhyphen="2D
\DeclareMathOperator*{\argmin}{arg\,min}
\DeclareMathOperator*{\argmax}{arg\,max}
\makeatletter
\newcommand{\myl}{\bBigg@{0.83}}
\makeatother

\renewcommand{\qed}{$\blacksquare$}
\newcommand{\floor}[1]{\left\lfloor {#1} \right\rfloor}

\renewcommand{\le}{\leqslant}

\renewcommand{\ge}{\geqslant}

\newcommand{\npcomplete}{$\mathcal{NP}$-complete\xspace}
\newcommand{\nphard}{$\mathcal{NP}$-hard\xspace}

\newcommand{\bbN}{\mathbb{N}}
\newcommand{\bbR}{\mathbb{R}}

\newif\ifcomments
\commentstrue
\ifcomments
\newcommand{\kibitz}[2]{{\color{#1}{#2}}}
\else
\newcommand{\kibitz}[2]{}
\fi

\newcommand{\vc}[1]{\kibitz{red}{[VC: #1]}}

\renewcommand{\vec}{\mathbf}

\newcommand{\seqdecision}{public decision making\xspace}
\newcommand{\issue}{issue\xspace}
\newcommand{\issues}{issues\xspace}
\newcommand{\mms}{\text{MMS}\xspace}
\newcommand{\rrs}{\text{RRS}\xspace}
\newcommand{\pps}{\text{PPS}\xspace}
\newcommand{\prop}{\text{Prop}\xspace}
\newcommand{\propone}{\text{Prop1}\xspace}
\renewcommand{\vc}{\vec{c}}
\newcommand{\vA}{\vec{A}}

\newcommand{\set}[1]{\{{#1}\}}

\begin{document}
\title{Fair Public Decision Making\thanks{A preliminary version of this paper appears in EC'17.}}  
\author{Vincent Conitzer \\ Duke University \\ \texttt{conitzer@cs.duke.edu}
\and
Rupert Freeman \\ Duke University \\ \texttt{rupert@cs.duke.edu}
\and
Nisarg Shah \\ Harvard University \\ \texttt{nisarg@g.harvard.edu}}
  
\date{}

\maketitle

\begin{abstract}
We generalize the classic problem of fairly allocating indivisible goods to the problem of \emph{fair public decision making}, in which a decision must be made on several social issues simultaneously, and, unlike the classic setting, a decision can provide positive utility to multiple players. We extend the popular fairness notion of proportionality (which is not guaranteeable) to our more general setting, and introduce three novel relaxations --- \emph{proportionality up to one issue, round robin share, and pessimistic proportional share} --- that are also interesting in the classic goods allocation setting. We show that the Maximum Nash Welfare solution, which is known to satisfy appealing fairness properties in the classic setting, satisfies or approximates all three relaxations in our framework. We also provide polynomial time algorithms and hardness results for finding allocations satisfying these axioms, with or without insisting on Pareto optimality. 
\end{abstract}


\input{1-intro}
\input{2-model}

\input{3-mnw}
\input{4-comp}
\input{5-disc}

\section*{Acknowledgements}
	Conitzer and Freeman are supported by NSF under awards
	IIS-1527434 and CCF-1337215, ARO under grants W911NF-12-1-0550 and
	W911NF-11-1-0332, and a Guggenheim Fellowship.


\clearpage
\appendix 
\section*{Appendix}
\input{6-appendix}

\end{document}

%% file: 1-intro.tex
\newcommand{\pgd}{private goods division\xspace}
\section{Introduction}
\label{sec:intro}

The literature on mathematically rigorous fair division dates back to the work of \citet{Stein48}. In the field's long history, most work focuses on the fair division of \emph{private goods}, in which a set of $m$ items must be divided among a set of $n$ players. Agents express their preferences by specifying their value for each good, and our goal is to find a division of the goods that is fair to all players.

One particularly appealing notion of fairness is \emph{envy-freeness}~\cite{Fol67}, which says that no player should want to switch her set of items with that of another player. This is a natural and strong notion of fairness that has long been the subject of fair division research~\cite{Var74,RP98,Klijn00,HRS02,BL08,BKK12,CKKK12}. It actually implies many other fairness notions such as \emph{proportionality}~\citep{Stein48} --- each player should get at least a $1/n$ fraction of her value for the entire set of goods --- and \emph{envy-freeness up to one good (EF1)}~\citep{LMMS04} --- no player should envy another player after removing at most one good from the latter player's bundle. Unfortunately, envy-freeness cannot always be guaranteed, and therefore its relaxations have been focused on~\cite{LMMS04,Bud11,PW14,CKMP+16}. 

Division of private goods, however, is not the only application in which we may desire a fair outcome. Often, we may need to make decisions where every alternative gives positive utility to many players, rather than to just one player as in the case of private goods. For instance, consider a couple, Alice and Bob, deciding where to go to dinner. Alice likes Italian food the most, but does not like Indian, whereas Bob prefers Indian food but does not like Italian. 
When there is only a single decision to make, we are simply in a classic bargaining game where players must attempt to arrive at a mutually agreeable solution. \citet{Nash50b} proposed maximizing the product of players' utilities (the \emph{Nash welfare}) as an elegant solution that uniquely satisfies several appealing properties. But no matter how we arrive at a decision -- and there is a myriad of work in computational social choice~\cite{BCEL+16} discussing how exactly we should do so -- some tradeoff must necessarily be made, and we may not be able to make everyone happy.


However, if we have several public decisions to make, maybe we can reach a compromise by making sure that all players are happy with at least some of the decisions. For example, if Alice and Bob are to follow their dinner with a movie, then maybe Bob will be willing to eat Italian food for dinner if he gets to pick his favorite movie, and maybe Alice will agree to this compromise. 

Note that this setting generalizes the classic private goods setting, because in this special case we can view each public decision as the allocation of a single good. While envy is a compelling notion in the private goods setting, it makes less sense for public decisions. 
In our example, irrespective of where Alice and Bob go for dinner, because they are eating the same food, it is not clear what it would mean for Alice to envy Bob. If she could somehow trade places with Bob, she would still be sitting at the other end of the dinner table, eating the same food, and not be any better off. 
Thankfully, proportionality still has a sensible interpretation: Each player should get at least a $1/n$ fraction of the utility she would get if her most desired alternative was chosen for each decision. Unfortunately, as with envy-freeness, proportionality cannot always be guaranteed. Therefore in this work we consider relaxations of proportionality in order to arrive at fairness notions that can be guaranteed. 


\subsection{Our Results}
\label{subsec:results}

Formally, a public decision making problem consists of $m$ issues, where each issue has several associated alternatives. Each of $n$ players has a utility for each alternative of each issue. Making a decision on an issue amounts to choosing one of the alternatives associated with the issue, and choosing an overall outcome requires making a decision on each issue simultaneously. The utility to a player for an outcome is the sum of her utilities for the alternatives chosen for different issues.
This is a very simple setting, but one in which the problem of fairness is already non-trivial.

We propose relaxations of proportionality in two directions. The first, \emph{proportionality up to one issue (\propone)}, is similar in spirit to EF1, stating that a player should be able to get her proportional share if she gets to change the outcome of a single issue in her favor. The second direction is based on the guarantees provided by the round robin mechanism. This mechanism first orders the set of players, and then repeatedly goes through the ordering, allowing each player to make her favorite decision on any single issue, until decisions are made on all the issues. Our first relaxation in this direction, the \emph{round robin share (\rrs)}, guarantees each player the utility that she would have received under the round robin mechanism if she were the last player in the ordering. Note that the round robin mechanism lets each player make decisions on roughly the same number of issues. A further relaxation in this direction, the \emph{pessimistic proportional share (\pps)}, guarantees each player the utility that she would get if her favorite alternatives were chosen for (approximately) a $1/n$ fraction of the issues, where these issues are chosen adversarially. 

We examine the possibility and computational complexity of satisfying combinations of these fairness desiderata. We first observe that the round robin mechanism satisfies both \propone and \rrs (and thus \pps). However, it fails to satisfy even the most basic efficiency property, Pareto optimality (PO), which requires that no other outcome should be able to make a player strictly better off without making at least one player strictly worse off. 

	{\small
		\begin{table}[ht]
			\begin{center}
				\begin{tabular}{c|cccc}
					& PO & \pps & \rrs & \propone\\
					\hline
					MNW, Private goods & \checkmark & \checkmark~(Th.~\ref{thm:mnw-rrs-pps-pgd}) & $\frac{1}{2}$~(Th.~\ref{thm:mnw-rrs-pps-pgd}) & \checkmark~(Th.~\ref{thm:mnw-prop1})\\[0.1cm]
					MNW, Public decisions & \checkmark~(Th.~\ref{thm:mnw-prop1}) & $\frac{1}{n}$~(Th.~\ref{THM:MNW-RRS-PPS}) & $\frac{1}{n}$~(Th.~\ref{THM:MNW-RRS-PPS}) & \checkmark~(Th.~\ref{thm:mnw-prop1})\\[0.1cm]
					Leximin Mechanism & \checkmark & \checkmark~(Th.~\ref{thm:leximin}) & \checkmark~(Th.~\ref{thm:leximin}) & $\frac{1}{2}$~(Th.~\ref{thm:leximin}) \\[0.1cm]
					Round Robin Method & $\times$ & \checkmark~(Th.~\ref{thm:round-robin}) & \checkmark~(Th.~\ref{thm:round-robin})  & \checkmark~(Th.~\ref{thm:round-robin}) \\
				\end{tabular}
			\end{center} 
			\caption{Axioms satisfied or approximated by the mechanisms we consider. The MNW solution is split into private goods and general decisions because we obtain significantly stronger results for private goods. Results for the leximin mechanism and the round robin method apply equally to private goods and public decisions. The approximation results are lower bounds; we omit the upper bounds from the table for simplicity.}
			\label{tab:summary}
		\end{table}}

When insisting on Pareto optimality, we observe that the leximin mechanism --- informally, it chooses the outcome that maximizes the minimum utility to any player --- satisfies \rrs (therefore \pps) and PO via a simple argument. However, this argument does not extend to establishing \propone, although we show that \rrs implies a $1/2$ approximation to \propone. To that end, we prove that the maximum Nash welfare (MNW) solution --- informally, it chooses the outcome that maximizes the \emph{product} of utilities to players --- that is known for its many desirable fairness properties in dividing private goods~\citep{CKMP+16} satisfies \propone and PO in our public decision making framework, and simultaneously provides a $1/n$ approximation to both \rrs and \pps. We also show that this approximation is tight up to a factor of $O(\log n)$. For division of private goods, these approximations improve significantly: the MNW solution completely satisfies \pps, and provides an $n/(2n-1) > 1/2$ approximation (but not better than $2/3$ approximation in the worst case) to \rrs. Table~\ref{tab:summary} provides a summary of these results.

However, both the MNW outcome and the leximin outcome are \nphard to compute. It is therefore natural to consider whether our fairness properties can be achieved in conjunction with PO in polynomial time. For public decision making, the answer turns out to be negative for \pps and \rrs, assuming $\mathcal{P} \neq \mathcal{NP}$. For division of private goods, however, we show that there exists a polynomial time algorithm that satisfies \pps and PO. 

\subsection{Related Work}
\label{subsec:related}

Two classic fair division mechanisms --- the leximin mechanism and the maximum Nash welfare (MNW) solution --- play an important role in this paper. Both mechanisms have been extensively studied in the literature on \pgd. In particular, \citet{KPS15} (Section 3.2) show that the leximin mechanism satisfies envy-freeness, proportionality, Pareto optimality, and a strong game-theoretic notion called group strategyproofness, which prevents even groups of players from manipulating the outcome to their benefit by misrepresenting their preferences, in a broad fair division domain with private goods and a specific form of non-additive utilities. On the other hand, the MNW solution has been well studied in the realm of additive utilities~\cite{RE10,CKMP+16}. For divisible goods, the MNW solution coincides with another well-known solution concept called \emph{competitive equilibrium from equal incomes} (CEEI)~\cite{Var74}, which also admits an approximate version for indivisible goods~\cite{Bud11}. For indivisible goods, the MNW solution satisfies envy-freeness up to one good, Pareto optimality, and approximations to other fairness guarantees. One line of research aims to approximate the optimum Nash welfare~\cite{CG15,Lee15}, although it is unclear if this achieves any of the appealing fairness guarantees of the MNW solution.

Our model is closely related to that of voting in combinatorial domains (see~\cite{LX16} for an overview). 
However, this literature focuses on the case where there is dependency between decisions on different issues. In contrast, our model remains interesting even though the issues are independent, and incorporating dependency is an interesting future direction.
Although there is a range of work in the voting literature that focuses on fairness~\cite{CC83,Mon95,BMS05,ABCE+17}, especially in the context of \emph{representation} in multi-winner elections, it focuses on ordinal, rather than cardinal, preferences.\footnote{That said, there is a recent line of work on \emph{implicit utilitarian voting} that attempts to maximize an objective with respect to the cardinal utilities underlying the ordinal preferences~\cite{PR06,BCHL+15}, and is therefore closer to our work.} Another difference is that fairness concepts in voting apply most naturally when $n>>m$, whereas our notions of fairness are most interesting when $m \ge n$.


Our work is also reminiscent of the \emph{participatory budgeting} problem~\cite{Cab04}, in which there are multiple public projects with different costs, and a set of projects need to be chosen based on preferences of the participants over the projects, subject to a budget constraint. Recently, researchers in computational social choice have addressed this problem from an axiomatic viewpoint~\cite{GKSA16}, including fairness considerations~\cite{FGM16}, and from the viewpoint of \emph{implicit utilitarian voting}~\cite{BPNS17}. However, they assume access only to ordinal preferences (that may stem from underlying cardinal utilities), while we assume a direct access to cardinal utilities, as is common in the fair division literature. Also, we do not have a budget constraint that binds the outcomes on different issues. 

%% file: 2-model.tex
\newcommand{\para}[1]{\smallskip\noindent\textbf{#1}}
\newcommand{\calC}{\mathcal{C}}
\section{Model}
\label{sec:prelim}

For $k \in \bbN$, define $[k] \triangleq \{1,\ldots,k\}$. Before we introduce the problem we study in this paper, let us review the standard fair division setting with private goods.

\para{Private goods division.} A \pgd problem consists of a set of \emph{players} $N=[n]$ and a set of $m$ \emph{goods} $M$. Each player $i \in N$ is endowed with a utility function $u_i : M \to \bbR_{+}$ such that $u_i(g)$ denotes the value player $i$ derives from good $g \in M$. A standard assumption in the literature is that of additive valuations, i.e., (slightly abusing the notation) $u_i(S) = \sum_{g \in S} u_i(g)$ for $S \subseteq M$. An \emph{allocation} $\vA$ is a partition of the set of goods among the set of players, where $A_i$ denotes the bundle of goods received by player $i$. Importantly, players only derive utility from the goods they receive, i.e., the goods private to them. The utility of player $i$ under allocation $\vA$ is $u_i(\vA) = u_i(A_i)$. 

\para{Public decision making.} A \emph{\seqdecision} problem also has a set of players $N = [n]$, but instead of private goods, it has a set of \emph{\issues} $T = [m]$. Each \issue $t \in T$ has an associated set of alternatives $A^t=\{a_1^t, \ldots, a_{k_t}^t\}$, exactly one of which must be chosen. Each player $i$ is endowed with a utility function $u_i^t : A^t \to \bbR_{+}$ for each \issue $t$, and derives utility $u_i^t(a_j^t)$ if alternative $a_j^t$ is chosen for \issue $t$. In contrast to private goods division, a single alternative can provide positive utility to multiple players. 

An outcome $\vc=(c_1, \ldots, c_m)$ of a \seqdecision problem is a choice of an alternative for every \issue, i.e., it consists of an outcome $c_t \in A^t$ for each issue $t \in T$. Let $\calC$ denote the space of possible outcomes. Slightly abusing the notation, let $u_i^t(\vc) = u_i^t(c_t)$ be the utility player $i$ derives from the outcome of \issue $t$. We also assume additive valuations: let $u_i(\vc) = \sum_{t \in T} u_i^t(\vc)$ be the utility player $i$ derives from outcome $\vc$.

In this work, we study deterministic outcomes, and in Section~\ref{sec:disc}, discuss the implications when randomized outcomes are allowed. Further, we study the \emph{offline} problem in which we are presented with the entire problem up front, and need to choose the outcomes on all \issues simultaneously. 
One can also define an online version of the problem~\cite{FZC17}, in which we must commit to the outcome of \issue $t$ before observing \issues $t'$ with $t' > t$ for all $t$, but we do not consider that version here.

\para{Private goods versus public decisions.} To see why \seqdecision generalizes \pgd, take an instance of \pgd, and create an instance of \seqdecision as follows. Create an issue $t_g$ for each good $g$. Let there be $n$ alternatives in $A^{t_g}$, where alternative $a_i^{t_g}$ gives player $i$ utility $u_i(g)$ while giving zero utility to all other players. It is easy to see that choosing alternative $a_i^{t_g}$ is equivalent to allocating good $g$ to player $i$. Hence, the constructed \seqdecision problem effectively mimics the underlying \pgd problem. 

%
%

\subsection{Efficiency and Fairness} 

In this paper, we not only adapt classical notions of efficiency and fairness defined for \pgd to our \seqdecision problem, but also introduce three fairness axioms that are novel for both \seqdecision and \pgd. First, we need additional notation that we will use throughout the paper. 

Let $p \triangleq \floor{m/n}$. For \issue $t \in T$ and player $i \in N$, let $a_{max}^t(i) \in \argmax_{a \in A^t} \{ u_i^t(a) \}$ and $u^t_{max}(i)=u_i^t(a_{max}^t(i))$. That is, $a_{max}^t(i)$ is an alternative that gives player $i$ the most utility for \issue $t$, and $u^t_{max}(i)$ is the utility player $i$ derives from $a_{max}^t(i)$. Let the sequence $\langle u^{(1)}_{max}(i),\ldots,u^{(m)}_{max}(i)\rangle$ represent the maximum utilities player $i$ can derive from different \issues, sorted in a non-ascending order. Hence, $\{u^{(k)}_{max}(i)\}_{k \in [m]} = \{u^t_{max}(i)\}_{t \in T}$ and $u^{(k)}_{max}(i) \ge u^{(k+1)}_{max}(i)$ for $k \in [m-1]$. 

\para{Efficiency.} In this paper, we focus on a popular notion of economic efficiency. We say that an outcome $\vc$ is \emph{Pareto optimal} (PO) if there does not exist another outcome $\vec{c'}$ that can provide at least as much utility as $\vc$ to every player, i.e., $u_i(\vec{c'}) \ge u_i(\vc)$ for all $i \in N$, and strictly more utility than $\vc$ to some player, i.e., $u_{i^*}(\vec{c'}) > u_{i^*}(\vc)$ for some $i^* \in N$.

\para{Fairness.} For \pgd, perhaps the most prominent notion of fairness is envy-freeness~\cite{Fol67}. An allocation $\vA$ is called \emph{envy-free} (EF) if every player values her bundle at least as much as she values any other player's bundle, i.e., $u_i(A_i) \ge u_i(A_j)$ for all $i,j \in N$. Because envy-freeness cannot in general be guaranteed, prior work also focuses on its relaxations. For instance, an allocation $\vA$ is called \emph{envy-free up to one good} (EF1) if no player envies another player after removing at most one good from the latter player's bundle, i.e., for all $i,j \in N$, either $u_i(A_i) \ge u_i(A_j)$ or $\exists g_j \in A_j$ such that $u_i(A_i) \ge u_i(A_j\setminus\set{g_j})$. 

Unfortunately, as argued in Section~\ref{sec:intro}, the notion of envy is not well defined for public decisions. Hence, for \seqdecision, we focus on another fairness axiom, Proportionality, and its relaxations. For \pgd, proportionality is implied by envy-freeness.\footnote{This assumes non-wastefulness, i.e., that all goods are allocated. We make this assumption throughout the paper.}

\para{Proportionality (\prop).} At a high level, proportionality requires that each player must receive at least her ``proportional share'', which is a $1/n$ fraction of the utility she would derive if she could act as the dictator. For a \seqdecision problem, the \emph{proportional share} of player $i$ ($\prop_i$) is $1/n$ times the sum of the maximum utilities the player can derive across all \issues, i.e., 
\[ \prop_i = \frac{1}{n} \sum_{t \in T} u^t_{max}(i) .\]
For $\alpha \in (0,1]$, we say that an outcome $\vc$ satisfies $\alpha$-\emph{proportionality} ($\alpha$-\prop) if $u_i(\vc) \ge \alpha \cdot \prop_i$ for all players $i \in N$. We refer to $1$-\prop simply as \prop. 

\para{Proportionality up to one \issue (\propone).} We introduce a novel relaxation of proportionality (more generally, of $\alpha$-proportionality) in the same spirit as envy-freeness up to one good, which is a relaxation of envy-freeness. For $\alpha \in (0,1]$, we say that an outcome $\vc$ satisfies $\alpha$-\emph{proportionality up to one \issue} ($\alpha$-\propone) if for every player $i \in N$, there exists an \issue $t \in T$ such that, ceteris paribus, changing the outcome of $t$ from $c_t$ to $a^t_{max}(i)$ ensures that player $i$ achieves an $\alpha$ fraction of her proportional share, i.e., if 
\[ 
\forall i \in N\ \ \exists t \in T\ \text{ s.t. }\ u_i(\vc) - u_i^t(\vc) + u^t_{max}(i) \ge \alpha \cdot \prop_i.
\] 
We refer to $1$-\propone simply as \propone.

\para{Round robin share (\rrs).} Next, we introduce another novel fairness axiom that is motivated from the classic round robin method that, for private goods, lets players take turns and in each turn, pick a single most favorite item left unclaimed. For \seqdecision, we instead let players make a decision on a single \issue in each turn. The utility guaranteed to the players by this approach is captured by the following fairness axiom. 


Recall that the sequence $\langle u^{(1)}_{max}(i),\ldots,u^{(m)}_{max}(i) \rangle$ represents the maximum utility player $i$ can derive from different \issues, sorted in a non-ascending order. Then, we define the \emph{round robin share} of player $i$ ($\rrs_i$) as 
\[ 
\rrs_i = \sum_{k=1}^{p} u^{(k\cdot n)}_{max}(i).
\]
This is player $i$'s utility from the round robin method, if she is last in the ordering and all \issues she does not control give her utility 0.
For $\alpha \in (0,1]$, we say that an outcome $\vc$ satisfies $\alpha$-\emph{round robin share} ($\alpha$-\rrs) if $u_i(\vc) \ge \alpha \cdot \rrs_i$ for all players $i \in N$. Again, we refer to $1$-\rrs simply as \rrs.

\para{Pessimistic proportional share (\pps).} We introduce another novel fairness axiom that is a further relaxation of round robin share. Note that the round robin method, by letting players make a decision on a single \issue per turn, allows each player to make decisions on at least $p = \floor{m/n}$ \issues. The following axiom captures the utility players would be guaranteed if each player still made decisions on a ``proportional share'' of $p$ issues, but if these issues were chosen pessimistically. 

We define the \emph{pessimistic proportional share} of player $i$ ($\pps_i$) to be the sum of the maximum utilities the player can derive from a set of $p$ \issues, chosen adversarially to minimize this sum:
\[ 
\pps_i = \sum_{k=m-p+1}^m u^{(k)}_{max}(i).
\]
For $\alpha \in (0,1]$, we say that an outcome $\vc$ satisfies $\alpha$-\emph{pessimistic proportional share} ($\alpha$-\pps) if $u_i(\vc) \ge \alpha \cdot \pps_i$ for all players $i \in N$. Again, we refer to $1$-\pps simply as \pps.

\para{Connections among fairness properties.} Trivially, proportionality (\prop) implies proportionality up to one \issue (\propone). In addition, it can also be checked that the following sequence of logical implications holds: $\prop \implies \mms \implies \rrs \implies \pps$.

Here, \mms is the \emph{maximin share guarantee}~\cite{Bud11,PW14}. Adapting the definition naturally from \pgd to \seqdecision, 
the maximin share of a player is the utility the player can guarantee herself by dividing the set of \issues into $n$ bundles, if she gets to make the decisions best for her on the \issues in an adversarially chosen bundle. 
The maximin share (\mms) guarantee requires that each player must receive utility that is at least her maximin share. We do not focus on the maximin share guarantee in this paper. 

\subsection{Mechanisms}

A mechanism for a \seqdecision problem (resp. a \pgd problem) maps each input instance of the problem to an outcome (resp. an allocation). We say that a mechanism satisfies a fairness or efficiency property if it always returns an outcome satisfying the property. There are three prominent mechanisms that play a key role in this paper. 

\para{Round robin method.} As mentioned earlier, the round robin method first fixes an ordering of the players. Then the players take turns choosing their most preferred alternative on a single issue of their choice whose outcome has not yet been determined. 

\para{The leximin mechanism.} The leximin mechanism chooses an outcome which maximizes the utility of the worst off player, i.e., $\min_{i \in N} u_i(\vc)$. Subject to this constraint, it maximizes the utility of the second least well off player, and so on. Note that the leximin mechanism is trivially Pareto optimal because if it were possible to improve some player's utility without reducing that of any other, it would improve the objective that the leximin mechanism optimizes. 

\para{Maximum Nash welfare (MNW).} The \emph{Nash welfare} of an outcome $\vc$ is the product of utilities to all players under $\vc$: $NW(\vc)=\prod_{i \in N} u_i(\vc)$. When there exists an outcome $\vc$ with $NW(\vc)>0$, then the MNW solution chooses an arbitrary outcome $\vc$ that maximizes the Nash welfare. When all outcomes have zero Nash welfare, it finds a largest cardinality set $S$ of players that can be given non-zero utility, and selects an outcome maximizing the product of their utilities, i.e., $\prod_{i \in S} u_i(\vc)$. 

\subsection{Examples}

We illustrate the fairness properties through two examples. 

\begin{example}
	\label{ex:ex1}
	Consider a \seqdecision problem with two players ($N=[2]$) and two \issues ($T=[2]$). Each \issue has two alternatives ($|A^1| = |A^2| = 2$). The utilities of the two players for the two alternatives in both \issues are as follows. 
	\medskip
	\begin{center}
	\begin{tabular}{c|cc}
	&$a_1^t$&$a_2^t$\\
	\hline
	$u^t_1$ & 1 & 0\\
	$u^t_2$ & 0 & 1\\
	\end{tabular} \quad for $t \in [2]$.
	\end{center} 
	\medskip
	The various fair shares of the two players are $\prop_1 = \rrs_1 = \pps_1 = \prop_2 = \rrs_2 = \pps_2 = 1$. Now, outcome $\vc=(a_1^1, a_1^2)$ gives utilities $u_1(\vc) = 2$ and $u_2(\vc) = 0$, and therefore violates \prop, \rrs, and \pps. It satisfies \propone because switching the decision on either \issue in favor of player $2$ makes her achieve her proportional share. On the other hand, outcome $\vc = (a^1_1,a^2_2)$ gives utility $1$ to both players, and thus satisfies \prop (as well as \propone, \rrs, and \pps, which are relaxations of \prop). 
\end{example}

\begin{example}
	\label{ex:ex2}
	Consider a \seqdecision problem with two players ($N=[2]$) and eight \issues ($T=[8]$). Once again, each \issue has two alternatives, for which the utilities of the two players are as follows. 
	\medskip
	\begin{center}
	\begin{tabular}{c|cc}
	&$a_1^t$&$a_2^t$\\
	\hline
	$u^t_1$ & 1 & 0\\
	$u^t_2$ & 0 & 1\\
	\end{tabular}
	\quad for $t \in \{1,2,3,4\}$, \quad and \quad
	\begin{tabular}{c|cc}
	&$a_1^t$&$a_2^t$\\
	\hline
	$u^t_1$ & 1&0\\
	$u^t_2$ & 0&0\\
	\end{tabular}
	\quad for $t \in \{5,6,7,8\}$.
	\end{center}
	\medskip
	In this case, we have $\prop_1 = \rrs_1 = \pps_1 = 4$, whereas $\prop_2 = \rrs_2 = 2$ and $\pps_2 = 0$. Consider outcome $\vc=(a_1^1, a_1^2, a_1^3, a_1^4, a_1^5, a_1^6, a_1^7, a_1^8)$. Then, we have $u_1(\vc) = 8$ while $u_2(\vc) = 0$, which satisfies \pps but violates \rrs. Further, $\vc$ also violates \propone because switching the outcome of any single \issue can only give player $2$ utility at most $1$, which is less than $\prop_2 = 2$. On the other hand, outcome $\vc = (a_2^1,a_2^2,a_2^3,a_2^4,a_1^5,a_1^6,a_1^7,a_1^8)$ achieves $u_1(\vc) = u_2(\vc) = 4$, and satisfies \prop (and thus its relaxations \propone, \rrs, and \pps). 
\end{example}


%% file: 3-mnw.tex
\section{(Approximate) Satisfiability of Axioms}
\label{sec:mnw}

If we are willing to sacrifice Pareto optimality, then we can easily achieve both \rrs (and therefore \pps) and \propone simultaneously with the round robin mechanism. This is not a surprising result. \rrs is defined based on the guarantee provided by the round robin mechanism, and \pps is a relaxation of \rrs. The round robin mechanism is also known to satisfy EF1 for \pgd, which is similar in spirit to \propone. 

\begin{theorem}
The round robin mechanism satisfies \rrs (and therefore \pps) and \propone, and runs in polynomial time.
\label{thm:round-robin}
\end{theorem}
\begin{proof}
	The round robin mechanism clearly runs in polynomial time (note that it is easy for a player to choose the next \issue on which to determine the outcome). To see why it satisfies \rrs, note that the mechanism allows every player to make a decision on one issue once every $n$ turns. Thus, for each $k \in [p]$, every player gets to make decisions on at least $k$ of her ``top'' $k \cdot n$ issues, when issues are sorted in the descending order of the utility her favorite alternative in the issue gives her.  It is easy to see that this implies every player $i$ gets utility at least $\rrs_i$. Because \rrs implies \pps, the mechanism also satisfies \pps. It remains to show that it satisfies \propone as well.
	
	Fix a player $i$ and let $\vc$ be the outcome produced by the round robin mechanism for some choosing order of the players. Because the round robin mechanism satisfies \rrs, player $i$ gets utility at least 
	\[
		u_i(\vc) \ge \sum_{k=1}^{p} u^{(k\cdot n)}_{max}(i).
	\]
	
	
	For $k \in [m]$, let the $k^{\text{th}}$ favorite issue of player $i$ be the issue $t$ for which $u^t_{max}(i)$ is the $k^{\text{th}}$ highest. Let $\ell \in \bbN \cup \set{0}$ be the largest index such that for every $k \in [\ell]$, outcome $\vc$ chooses player $i$'s most preferred alternative on her $k^{\text{th}}$ favorite issue. Let $t^*$ be her $(\ell+1)^{\text{th}}$ favorite issue. To show that $\vc$ satisfies \propone, we construct outcome $\vec{c'}$ from $\vc$ by only changing the outcome of issue $t^*$ to $a^{t^*}_{max}(i)$, and show that $u_i(\vc') \ge \prop_i$.
	Note that if $\ell \ge p$, then 
	\[
	u_i(\vc') \ge \sum_{k=1}^{\ell +1} u^{(k)}_{max}(i) \ge \sum_{k=1}^{p+1} u^{(k)}_{max}(i) \ge \frac{1}{n} \sum_{k=1}^m u^{(k)}_{max}(i) = \prop_i. 
	\]
	Let $\ell < p$. Then, using the fact that the round robin mechanism lets player $i$ choose her most preferred alternative for at least $k$ of her favorite $k \cdot n$ issues for every $k \le p$ (and her $(\ell + 1)^{\text{th}}$ favorite \issue was not one of these), we have 
	\begin{align*}
	u_i(\vc') &\ge \sum_{k=1}^{\ell+1} u^{(k)}_{max}(i) + \sum_{k=\ell+1}^{p} u^{(k \cdot n)}_{max}(i) \\
		&\ge \frac{1}{n} \sum_{k=1}^{(\ell +1) \cdot n} u^{(k)}_{max}(i) + \frac{1}{n} \ \sum_{k=(\ell+1)n+1}^{m} u^{(k)}_{max}(i)
		= \prop_i .
		\end{align*}
Therefore, the round robin mechanism satisfies $\propone$.
\end{proof}

While this result seems to reflect favorably upon the round robin mechanism, recall that it violates Pareto optimality even for \pgd. For \seqdecision, a simple reason for this is that the round robin mechanism, for each issue, chooses an alternative that is some player's favorite, while it could be unanimously better to choose compromise solutions that make many players happy. Imagine there are two players and two issues, each with two alternatives. The ``extreme'' alternative in each issue $i \in \set{1,2}$ gives utility $1$ to player $i$ but $0$ to the other, while the ``compromise'' alternative in each issue $i \in \set{1,2}$ gives utility $2/3$ to both players. It is clear that both players prefer choosing the compromise alternative in both issues to choosing the extreme alternative in both issues. Because such ``Pareto improvements'' which make some players happier without making any player worse off are unanimously preferred by the players, the round robin outcome becomes highly undesirable. We therefore seek mechanisms that provide fairness guarantees while satisfying Pareto optimality.


A natural question is whether there exists a mechanism that satisfies \rrs, \propone, and PO. An obvious approach is to start from an outcome that already satisfies \rrs and \propone (e.g., the round robin outcome), and make Pareto improvements until no such improvements are possible. While Pareto improvements preserve \rrs as the utilities to the players do not decrease, \propone can be lost as it depends on the exact alternatives chosen and not only on the utilities to the players. We leave it as an important open question to determine if \rrs, \propone, and PO can be satisfied simultaneously.

We therefore consider satisfying each fairness guarantee individually with PO. One can easily find an outcome satisfying \rrs and PO by following the aforementioned approach of starting with an outcome satisfying \rrs, and making Pareto improvements while possible. 
There is also a more direct approach to satisfying \rrs and PO. Recall that the leximin mechanism chooses the outcome which maximizes the minimum utility to any player, subject to that maximizes the second minimum utility, and so on. It is easy to see that this mechanism is always Pareto optimal. Now, let us normalize the utilities of all players such that $\rrs_i = 1$ for every player $i \in N$.\footnote{Players with zero round robin share can be incorporated via a simple extension to the argument.} Because the round robin mechanism gives every player $i$ utility at least $\rrs_i = 1$, it must be the case that the leximin mechanism operating on these normalized utilities must also give every player utility at least $1$, and thus produce an outcome that is both \rrs and PO. 

\begin{theorem}\label{thm:leximin}
	The leximin mechanism satisfies \rrs, PO, and $(1/2)$-\propone.
\end{theorem}

That leximin satisfies $(1/2)$-\propone follows directly from the following lemma, and noting that leximin satisfies \rrs.

\begin{lemma}
	\rrs implies $(1/2)$-\propone. 
	\label{lem:rrs-propone}
\end{lemma}
\begin{proof}
	Note that 
	\[\rrs_i = \sum_{k=1}^p u_{max}^{(k \cdot n)}(i) \ge \frac{1}{n} \sum_{t=n+1}^m u_{max}^{(t)}(i)\] 
	and 
	\[u_{max}^{(1)}(i) \ge \frac{1}{n} \sum_{t=1}^n u_{max}^{(t)}(i).\] 
	Summing the two equations, we get
	\[ \rrs_i + u_{max}^{(1)}(i) \ge \frac{1}{n} \sum_{t=1}^m u_{max}^{(t)}(i) = \prop_i. \]
	Therefore, $\max \{ \rrs_i, u_{max}^{(1)}(i) \} \ge \frac{1}{2} \prop_i$. 
	
	Suppose that $u_i(\vc) \ge \rrs_i$ for some outcome $\vc$. Then either $i$ already receives her most valued item, in which case she receives utility at least $\max \{ \rrs_i, u_{max}^{(1)}(i) \} \ge \frac{1}{2} \prop_i$, or she does not receive her most valued item. If she does not, then after giving it to her, she receives utility at least $\frac{1}{2} \prop_i$. Therefore, $\vc$ satisfies $(1/2)$-\propone.
	\end{proof}

Next, we study whether we can achieve \propone and PO simultaneously. Neither of the previous approaches seems to work: we already argued that following Pareto improvements could lose \propone, and the normalization trick is difficult to apply because \propone is not defined in terms of any fixed share of utility. 

One starting point to achieving \propone and PO is the maximum Nash welfare (MNW) solution, which, for \pgd, is known to satisfy the similar guarantee of EF1 and PO~\cite{CKMP+16}. It turns out that the MNW solution is precisely what we need. 

\begin{theorem}
\label{thm:mnw-prop1}
The MNW solution satisfies proportionality up to one issue (\propone) and Pareto optimality (PO).
\end{theorem}

Before we prove this, we need a folklore result, which essentially states that if the sum of $n$ terms is to be reduced by a fixed quantity $\delta$ that is less than each term, then their product reduces the most when $\delta$ is taken out of the lowest term. The following lemma proves this result when all initial terms are $1$, which is sufficient for our purpose. The proof of the lemma appears in the appendix. 

\begin{lemma}
Let $\set{x_1,\ldots,x_n}$ be a set of $n$ non-negative real numbers such that\\ $\sum_{i=1}^n \max \{ 0, 1- x_i\} \le \delta$, where $0 < \delta < 1$. Then, $\prod_{i=1}^n x_i \ge 1-\delta$.
\label{LEM:MAX-REDUCTION}
\end{lemma}

\begin{proof}[Proof of Theorem~\ref{thm:mnw-prop1}]
Fix an instance of the \seqdecision problem. Let $S \subseteq N$ be the set of players that the MNW outcome $\vc$ gives positive utility to. Then, by the definition of the MNW outcome, $S$ must be a largest set of players that can simultaneously be given positive utility, and $\vc$ must maximize the product of utilities of players in $S$. 

First, we show that $\vc$ is PO. Note that a Pareto improvement over $\vc$ must either give a positive utility to a player in $N\setminus S$ or give more utility to a player in $S$, without reducing the utility to any player in $S$. This is a contradiction because it violates either optimality of the size of $S$ or optimality of the product of utilities of players in $S$. Hence, MNW satisfies PO. 

We now show that MNW also satisfies \propone. Suppose for contradiction that \propone is violated for player $i$ under $\vc$. First, note that we must have $\prop_i > 0$. Further, it must be the case that $u^t_{max}(i)>0$ for at least $n+1$ \issues. Were this not the case, $\propone$ would be trivially satisfied for player $i$ since we can give her utility 
\[ u_{max}^{(1)}(i) \ge \frac{1}{n} \sum_{t=1}^n u_{max}^{(t)}(i) = \frac{1}{n} \sum_{t=1}^m u_{max}^{(t)}(i) = \prop _i \]
by changing the outcome on a single \issue. 

We now show that $u_i(\vc)>0$ (i.e., $i \in S$). For contradiction, suppose otherwise. For each of the (at least) $n+1$ \issues with $u^t_{max}(i)>0$, there must exist another player $j \not= i$ that gets positive utility \emph{only} from that issue under $\vc$ (otherwise we could use that \issue to give positive utility to $i$ while not reducing any other agents' utility to zero, contradicting the maximality of $S$). But this is impossible, since there are at least $n+1$ \issues and only $n-1$ agents (other than $i$).

Because MNW outcomes and the \propone property are invariant to individual scaling of utilities, let us scale the utilities such that $\prop_i = 1$ and $u_{j}(\vc)=1$ for all $j \in S\setminus\set{i}$. Select issue $t^* \in T$ as 
\[ 
t^* \in \argmin_{t \in T} \frac{\sum_{j \in N\setminus\set{i}} u_j^t(\vc)}{u^t_{max}(i)-u_i^t(\vc)}.
\]
Note that $t^*$ is well defined because $u^t_{max}(i) > u_i^t(\vc)$ for at least one $t \in T$, otherwise \propone would not be violated for player $i$.

We now show that outcome $\vec{c'}$ such that $c'_{t^*} = a^{t^*}_{max}(i)$ and $c'_t = c_t$ for all $t \in T\setminus\set{t^*}$ achieves strictly greater product of utilities of players in $S$ than outcome $\vc$ does, which is a contradiction as $\vc$ is an MNW outcome. First, note that
\begin{align}
 \frac{ \sum_{j \in N\setminus\set{i}} u_{j}^{t^*}(\vc) }{ u^{t^*}_{max}(i) - u_i^{t^*}(\vc) }
 \le \frac{ \sum_{t \in T} \sum_{j \in N\setminus\{i\}} u_{j}^t(\vc) }{ \sum_{t \in T} (u^t_{max}(i) - u_i^t(\vc)) } 
 = \frac{ \sum_{j \in N\setminus\set{i}} u_{j}(\vc) }{ n \prop_i - u_i(\vc) } 
 \le \frac{ (n-1) }{ (n-1)\prop_i } 
 = 1,
 \label{eqn:ineq-decrease}
\end{align}
where the penultimate transition follows because we normalized utilities to achieve $u_j(\vc) = 1$ for every $j \in S\setminus\set{i}$, every $j \in N\setminus S$ satisfies $u_j(\vc) = 0$, and player $i$ does not receive her proportional share. The final transition holds due to our normalization $\prop_i = 1$. 

Let $\delta = \sum_{j \in S\setminus\set{i}} u_j^{t^*}(\vc)$. Then, Equation~\eqref{eqn:ineq-decrease} implies $u_i(\vec{c'})-u_i(\vc) = u^{t^*}_{max}(i) - u^{t^*}_i(\vc) \ge \delta$. Thus, 
\begin{align}
u_i(\vc) + \delta \le u_i(\vec{c'}) < 1,
\label{eqn:i-increase}
\end{align}
where the last inequality follows because the original outcome $\vc$ violated \propone for player $i$. In particular, this implies $\delta < 1$. Our goal is to show that $\prod_{j \in S} u_j(\vec{c'}) > u_i(\vc) = \prod_{j \in S} u_j(\vc)$,
where the last equality holds due to our normalization $u_j(\vc) = 1$ for $j \in S\setminus\set{i}$ and because $i \in S$. This would be a contradiction because $\vc$ maximizes the product of utilities of players in $S$. Now, 
\begin{align*}
	 \sum_{j \in S \setminus \{ i \}} \max \{0, 1-u_j(\vc') \} = \sum_{j \in S \setminus \{ i \}} \max \{0, u_j^{t^*}(\vc)-u_j^{t^*}(\vc') \} \le  \sum_{j \in S \setminus \{ i \}} u_j^{t^*}(\vc) = \delta,
\end{align*}
where the first transition follows from setting $1=u_j(\vc)$ (by our normalization) and noting that $\vc$ and $\vc'$ are identical for all \issues except $t^*$, and the second because all utilities are non-negative.


Hence, Lemma~\ref{LEM:MAX-REDUCTION} implies that $\prod_{j \in S\setminus\set{i}} u_j(\vc') \ge 1-\delta$. Thus,
\[
\prod_{j \in S} u_j(\vc') \ge (1-\delta) \cdot (u_i(\vc)+\delta) = u_i(\vc) + \delta \cdot (1-u_i(\vc)) - \delta^2 > u_i(\vc) + \delta^2 - \delta^2 = u_i(\vc),
\] 
where the inequality holds because $u_i(\vc)+\delta < 1$ from Equation~\eqref{eqn:i-increase}.
%
%
\end{proof}

For \pgd, this result can be derived in a simpler fashion. \citet{CKMP+16} already show that MNW satisfies PO. In addition, they also show that MNW satisfies EF1, which implies \propone due to our next result. To be consistent with the goods division literature, we use proportionality up to one good (rather than one issue) in the \pgd context. 

\begin{lemma}
For \pgd, envy-freeness up to one good (EF1) implies proportionality up to one good (\propone).
\label{lem:pgd-ef1-prop1}
\end{lemma}
\begin{proof}
Take an instance of \pgd with a set of players $N$ and a set of goods $M$. Let $\vA$ be an allocation satisfying EF1. Fix a player $i \in N$. 

Due to the definition of EF1, there must exist\footnote{If $A_j = \emptyset$, we can add a dummy good $g_j$ that every player has utility $0$ for, and make $A_j = \set{g_j}$.} a set of goods $X = \set{g_j}_{j \in N\setminus\set{i}}$ such that $u_i(A_i) \ge u_i(A_j) - u_i(g_j)$ for every $j \in N\setminus\set{i}$. Summing over all $j \in N\setminus\set{i}$, we get
\begin{align}
(n-1) \cdot u_i(A_i) \ge \left(\sum_{j \in N\setminus\set{i}} u_i(A_j) \right)- u_i(X) &\implies n \cdot u_i(A_i) \ge u_i(M)-u_i(X) \nonumber \\
&\implies u_i(A_i) + \frac{u_i(X)}{n} \ge \frac{u_i(M)}{n}. \label{eqn:ef-implies-prop1}
\end{align}
Note that $X$ has less than $n$ goods. Suppose player $i$ receives good $g^* \in \argmax_{g \in X} u_i(g)$. Note that $g^* \notin A_i$. Then, we have $u_i(A_i \cup \set{g^*}) \ge u_i(M)/n = \prop_i$, which implies that \propone is satisfied with respect to player $i$. Because player $i$ was chosen arbitrarily, we have that EF1 implies \propone.
\end{proof}

Equation~\ref{eqn:ef-implies-prop1} in the proof of Lemma~\ref{lem:pgd-ef1-prop1} directly implies the following lemma because the set $X$ in the equation contains at most $n-1$ goods.

\begin{lemma} \label{lem:ef1}
	Let $\vA$ be an allocation of private goods that satisfies EF1. Then, for every player $i$, 
	\[ 
	u_i(A_i) \ge \prop_i - \frac{1}{n} \sum_{t=1}^{n-1} u^{(t)}_{max}(i), 
	\]
	where $u^{(t)}_{max}(i)$ is the utility player $i$ derives from her $t^{\text{th}}$ most valued good.
\end{lemma}

Next, we turn our attention to \rrs and \pps. While MNW does not satisfy either of them, it approximates both. 

\begin{theorem}
The MNW solution satisfies $1/n$-\rrs (and therefore $1/n$-\pps). The approximation is tight for both \rrs and \pps up to a factor of $O(\log n)$. 
\label{THM:MNW-RRS-PPS}
\end{theorem}

%

\begin{proof}
	We first show the lower bound. Fix an instance of \seqdecision, and let $\vc$ denote an MNW outcome. Let $S \subseteq N$ denote the set of players that achieve positive utility under $\vc$. 
	
	Without loss of generality, let us normalize the utilities such that $u_j(\vc)=1$ for every $j \in S$. Suppose for contradiction that for some player $i$, $u_i(\vc) < (1/n) \cdot \rrs_i$. First, this implies that $\rrs_i > 0$, which in turn implies that player $i$ must be able to derive a positive utility from at least $n$ different issues. By an argument identical to that used to argue that $u_i(\vc) > 0$ in the proof of Theorem~\ref{thm:mnw-prop1}, it can be shown that we must also have $u_i(\vc) > 0$ in this case (i.e., $i \in S$).
	
	Now, recall that the sequence $\langle u^{(1)}_{max}(i),\ldots,u^{(m)}_{max}(i) \rangle$ contains the maximum utility player $i$ can derive from different \issues, sorted in a non-ascending order. 
	For every $q \in [p]$, let 
	\[ t_q= \argmin_{(q-1)n+1 \le t \le qn } \sum_{j \in S \backslash \{ i \}} u_j^t(\vc). \]
That is, we divide the \seqdecision into sets of $n$ \issues, grouped by player $i$'s maximum utility for them, and for each set of \issues, we let $t_q$ be the one that the remaining players derive the lowest total utility from. Note that $t_q \le qn$ for each $q \in [p]$, and therefore $u^{(t_q)}_{max}(i) \ge u^{(qn)}_{max}(i)$. 

We will show that outcome $\vec{c'}$, where $c'_{t_q}=a^{t_q}_{max}(i)$ for all $q \in [p]$ and $c'_t=c_t$ for all other issues $t$, achieves a higher product of utilities to players in $S$ than $\vc$ does, which is a contradiction because $\vc$ is an MNW outcome.
First, note that 
\[ u_i(\vc ') \ge \sum_{q=1}^p u^{(t_q)}_{max}(i) \ge \sum_{k=1}^{p} u^{(k\cdot n)}_{max}(i) = \rrs_i > n.\]
Further, we have
\begin{align*}
	\sum_{j \in S \backslash \{ i \}} \max \{ 0, 1-u_j(\vc') \} &= \sum_{j \in S \backslash \{ i \}} \max \{ 0, u_j(\vc)-u_j(\vc') \}\\
	&= \sum_{j \in S \backslash \{ i \}} \sum_{q=1}^p \max \{ 0, u_j^{t_q}(\vc)-u_j^{t_q}(\vc') \}
	\le \sum_{j \in S \backslash \{ i \}} \sum_{q=1}^p u_j^{t_q}(\vc),
\end{align*}
	where the first equality follows from our normalization, the second because $\vc$ and $\vc'$ only differ on \issues $\set{t_q}_{q \in [p]}$, and the last because all utilities are non-negative.

Reversing the order of the summation and further manipulating the expression, we have
	\[ \sum_{q=1}^p \sum_{j \in S \backslash \{ i \}} u_j^{t_q}(\vc)
	\le \sum_{q=1}^p \frac{1}{n} \sum_{t=(q-1)n+1}^{qn}  \sum_{j \in S \backslash \{ i \}} u_j^t(\vc)
	= \frac{1}{n} \sum_{t=1}^{pn} \sum_{j \in S \backslash \{ i \}} u_j^t(\vc)
	\le \frac{n-1}{n}, \]
 where the first transition follows from the definition of $t_q$. By Lemma~\ref{LEM:MAX-REDUCTION}, we have
\[
\prod_{j \in S} u_j(\vc') = u_i(\vc') \prod_{j \in S \backslash \{ i \}} u_j(\vc') > n \cdot \left(1-\frac{n-1}{n}\right) =1 = \prod_{j=1}^n u_j(\vc),
\]
where the inequality holds because $u_i(\vec{c'}) \ge \rrs_i > n \cdot u_i(\vc) = n$, as player $i$ receives her round robin share under $\vec{c'}$ but did not even receive a $1/n$ fraction of it under $\vc$. Hence, outcome $\vc'$ achieves a higher product of utilities to players in $S$ than $\vc$ does, which is a contradiction.

For the upper bound, 
Consider a \seqdecision problem with $n$ \issues, where each \issue $t$ has two alternatives $a^t_1$ and $a^t_2$ with the following utilities to the players. Let $x=(\log n - \log \log n)/n$.

	\begin{center}
	\begin{tabular}{c|cc}
	&$a_1^1$&$a_2^1$\\
	\hline
	$u^1_1$ & 1 & $d$\\
	$u^1_2$ & 0 & x\\
	$\vdots$ & $\vdots$ & $\vdots$\\
	$u^1_n$ & 0 & x\\
	\end{tabular}
	\quad and \quad
	\begin{tabular}{c|cc}
	&$a_1^t$&$a_2^t$\\
	\hline
	$u^t_1$ & 1& $d$\\
	$u^t_2$ & 0 & 0\\
	$\vdots$ & $\vdots$ & $\vdots$\\
	$u^t_t$ & 0&1\\
	$\vdots$ & $\vdots$ & $\vdots$\\
	$u^t_n$ & 0 & 0
	\end{tabular}
	\quad for $t \in \{2,\ldots,n\}$.
	\end{center}
	We choose the value of $d$ later. Note that $\pps_1=1$ as long as $d<1$. Our goal is to make the MNW outcome choose alternative $a^t_2$ for every issue $t$. Let us denote this outcome by $\vc$. Then, the Nash welfare under $\vc$ is 
	\begin{equation}
		\label{eqn:mnw-pps-1}
		(n\cdot d) \cdot (1+x)^{n-1}.
	\end{equation}
	Let us find conditions on $d$ under which this is greater than the Nash welfare that other possible outcomes $\vc'$ could achieve. 
	
	Clearly, if $c'_1=a_1^1$ and $c'_t=a_1^t$ for any $t \neq 1$, then $u_t(\vc')=0$, and therefore $NW(\vc')=0$. Let us consider $\vc'$ under which $c'_1=a_1^1$ and $c'_t=a_2^t$ for all $t \neq 1$. The Nash welfare produced by this outcome is
	\begin{equation}
		\label{eqn:mnw-pps-2} 
		1+(n-1)d.
	\end{equation}
	
	Next, consider $\vc'$ with $c'_1=a_2^1$, $c'_{t^*}=a_1^{t^*}$ for some $t^* \neq 1$, and $c'_t = a_2^t$ for all remaining $t$. The Nash welfare under this outcome is
	\begin{equation}
		\label{eqn:mnw-pps-3}
		 (1+(n-1)d) \cdot x \cdot (1+x)^{n-2}.
	\end{equation}
	
	We do not need to consider outcomes $\vc'$ with $c'_1 = a_2^1$ and $c'_t = a_1^t$ for multiple $t \neq 1$. This is because if switching the outcome from $a_2^{t^*}$ to $a_1^{t^*}$ for even a single $t^* \neq 1$ decreases the Nash welfare, switching the outcomes on other $t \neq 1$ would only further decrease the Nash welfare, as it would reduce the utility to another player $t$ by the same factor $1/(1+x)$, while increasing the utility to player $1$ by an even smaller factor. 
	
	Let us identify the conditions on $d$ required for the quantity in Equation~\eqref{eqn:mnw-pps-1} to be greater than the quantities in Equations~\eqref{eqn:mnw-pps-2} and~\eqref{eqn:mnw-pps-3}. We need
	\begin{align}
		&\quad\ \  (n \cdot d) \cdot (1+x)^{n-1} > (1+(n-1)d) \cdot 1 \nonumber \\
		&\Leftrightarrow n \cdot d > \frac{1}{(1+x)^{n-1}-1+1/n}, \label{eqn:nash-pps-comparison-1}
	\end{align}
	and 
	\begin{align}
		&(n \cdot d) (1+x)^{n-1} > (1+(n-1)d) \cdot x \cdot (1+x)^{n-2} \nonumber\\
		&\Leftrightarrow n \cdot d > \frac{n \cdot x}{n+x}. \label{eqn:nash-pps-comparison-2}
	\end{align}
	
	It is easy to check that for $x = (\log n - \log \log n)/n$, the quantities on the RHS of both Equations~\eqref{eqn:nash-pps-comparison-1} and~\eqref{eqn:nash-pps-comparison-2} are $O(\log n/n)$. Hence, we can set $d$ to be sufficiently low for $n \cdot d$ to be $\Theta(\log n/n)$. However, note that $n \cdot d$ is precisely the approximation to \pps achieved for player $1$ under $\vc$, as required. 
\end{proof}	

For private goods, we can show that the MNW solution provides much better approximations to both \rrs and \pps, as a result of its strong fairness guarantee of EF1. 

\begin{lemma} 
\label{lem:ef1-pps-rrs}
For \pgd, envy-freeness up to one good (EF1) implies \pps and $n/(2n-1)$-\rrs, but does not imply $n/(2n-2)$-\rrs.
\end{lemma}

\begin{proof}
	Let $\vA$ be an allocation of private goods that satisfies EF1. First, we show that $\vA$ must also satisfy \pps. Suppose for contradiction that it violates \pps. Then, there exists a player $i$ such that $u_i(A_i)<\pps_i$, which in turn implies that $|A_i|<p$. Because the average number of goods per player is $\frac{m}{n} \ge p$, there must exist a player $j$ such that $|A_j|>p$. Hence, for any good $g \in A_j$, player $j$ has at least $p$ goods even after removing $g$ from $A_j$, which implies $u_i(A_j \backslash \{ g \}) \ge \pps_i > u_i(A_i)$. However, this contradicts the fact that $\vA$ is EF1. 

	 We now show that $\vA$ also satisfies $1/2$-\rrs. By Lemma~\ref{lem:ef1}, we have 
	 \begin{equation}  \label{eq:ef1-implies-1/2-rrs-1}
		  u_i(A_i) \ge \frac{1}{n} \sum_{t=n}^m u_{max}^{(t)}(i) \ge \frac{1}{n} u^{(n)}_{max}(i) + \sum_{k=2}^p u_{max}^{(k \cdot n)}(i).
	\end{equation}
Further, since $\vA$ satisfies EF1, it must be the case that 
\begin{equation} \label{eq:ef1-implies-1/2-rrs-2}
	u_i(A_i) \ge u_{max}^{(n)}(i).
\end{equation}
To see this, suppose for contradiction that $u_i(A_i) < u_{max}^{(n)}(i)$, which implies that player $i$ is not allocated any of her $n$ most valued goods. Therefore, by the pigeonhole principle, there must exist a player $j \in N\setminus\set{i}$ that is allocated at least two of these goods. Hence, for any $g \in A_j$, we have  
$u_i(A_j \setminus \{ g \}) \ge u_{max}^{(n)}(i) > u_i(A_i)$,
which violates EF1. Finally, adding $n$ times Equation~\eqref{eq:ef1-implies-1/2-rrs-1} with $n-1$ times Equation~\eqref{eq:ef1-implies-1/2-rrs-2}, we obtain
\[ 
(2n-1) \cdot u_i(A_i) \ge n\cdot u_{max}^{(n)}(i) + n\cdot \sum_{k=2}^p u_{max}^{(k \cdot n)}(i) = n\cdot \rrs_i,
\]
which implies the desired $n/(2n-1)$-\rrs guarantee.	 

For the upper bound, consider an instance with $n$ players and $n^2$ goods, and define player $1$'s utility function to be 
	 \[ u_1(g_j) = 
	 \begin{cases} 
		 1 & 1 \le j \le n, \\
		 \frac{1}{n-1} & n+1 \le j \le n^2.
		\end{cases} \]
		Note that $\rrs_1 = 1 + (n-1)\frac{1}{n-1}=2$.
		Consider the allocation $\vA$ with $A_1 = \{ g_{n+1}, \ldots, g_{2n} \}$, $A_2= \{g_1, g_2 \}$, and $A_i=\{ g_i, g_{(i-1)n+1}, \ldots, g_{i\cdot n} \}$ for all players $i >2$. Let the utilities of players $2$ through $n$ be positive for the goods they receive and zero for the remaining goods. Hence, they clearly do not envy any players. For player $1$, we have $u_1(A_1)=\frac{n}{n-1}$, $u_1(A_2 \setminus \{ g_2 \})=1$, and $u_1(A_i \setminus \{ g_i \})=\frac{n}{n-1}$ for all $i > 2$. That is, player $1$ does not envy any other player up to one good. Hence, $\vA$ satisfies EF1, and player 1 obtains a $\frac{n}{2n-2}$ fraction of her \rrs share, as required.
\end{proof}	
	
As a corollary of Lemma~\ref{lem:ef1-pps-rrs}, EF1 implies $1/2$-\rrs, and this approximation is asymptotically tight. Further, because the MNW solution satisfies EF1, Lemma~\ref{lem:ef1-pps-rrs} immediately provides guarantees (lower bounds) for the MNW solution. However, the upper bound in the proof of Lemma~\ref{lem:ef1-pps-rrs} does not work for the MNW solution. Next, we establish a much weaker lower bound, leaving open the possibility that the MNW solution may achieve a constant approximation better than $1/2$ to \rrs.

\begin{theorem}
For \pgd, the MNW solution satisfies \pps and $n/(2n-1)$-\rrs. For every $\eps > 0$, the MNW solution does not satisfy $(2/3+\epsilon)$-\rrs.
\label{thm:mnw-rrs-pps-pgd}
\end{theorem}
\begin{proof}
The lower bounds follow directly from Lemma~\ref{lem:ef1-pps-rrs} and the fact that the MNW solution satisfies EF1. For the upper bound, consider an instance with two players and four goods. Player $1$ has utilities $(1-\delta, 1-\delta, 1/2, 1/2)$ and player $2$ has utilities $(1,1,0,0)$ for goods $(g_1,g_2,g_3,g_4)$, respectively. Note that $\rrs_1=3/2-\delta$. The MNW allocation $\vA$ is given by $A_1=\{ g_3,g_4 \}$ and $A_2=\{ g_1, g_2 \}$. Thus, 
$\frac{u_1(A_1)}{\rrs_1}= \frac{2}{3-2\delta}$.
The upper bound follows by setting $\delta$ sufficiently small.
\end{proof}

%% file: 4-comp.tex
\section{Computational Complexity}
\label{sec:comp}

In Section~\ref{sec:mnw}, we showed that without requiring Pareto optimality, we can achieve both \rrs (thus \pps) and \propone in polynomial time using the round robin method (Theorem~\ref{thm:round-robin}). In contrast, when we require PO, the leximin mechanism (with an appropriate normalization of utilities) provides \rrs (thus \pps) and PO, while the MNW solution provides \propone and PO. However, both these solutions are \nphard to compute~\cite{RE10,BS06}. This raises a natural question whether we can efficiently find outcomes satisfying our fairness guarantees along with PO. For \pps, the answer is negative.

\begin{theorem} \label{thm:pps-po-np-hard}
It is \nphard to find an outcome satisfying \pps and PO.
\end{theorem}

Note that it is the \emph{search problem} of finding an outcome (any outcome) satisfying \pps and PO for which we prove computational hardness; the \emph{decision problem} of testing the existence of such an outcome is trivial as we know it always exists. Before we prove this result, we need to introduce a new (to our knowledge) decision problem and show that it is \npcomplete.

\paragraph{Exact Triple-Cover by 3-sets (X33C):} An instance $(Y,\mathcal{T})$ of X33C is given by a set $Y$ of $r$ vertices and a set $\mathcal{T} = \{ T_1, T_2, \ldots, T_m \}$, where each $T_i$ is a set of three vertices. The decision problem is to determine whether it is possible to choose $r$ sets, with repetition allowed, such that every vertex $v$ is contained in exactly three of the chosen sets (an exact triple-cover).

Let us contrast this with the definition of the popular NP-complete problem, Exact Cover by 3-sets (X3C): An instance $(X, \mathcal{S})$ of X3C is given by a set $X$ of $3q$ vertices and a set $\mathcal{S} = \{ S_1, \ldots, S_n \}$, where each $S_i$ is a set of three vertices. The decision problem is to determine if there exists a subset of $\mathcal{S}$ of size $q$ that covers every vertex $x \in X$ exactly once (an exact cover).

\begin{lemma} \label{LEM:X33X-NP-COMPLETE}
	X33C is NP-complete.
\end{lemma}

\begin{proof}
	Clearly, X33C lies in NP because a triple-cover can be checked in polynomial time. To show hardness, we reduce from X3C. Given an instance $(X, \mathcal{S})$ of X3C, divide $X$ into $q$ sets of 3 vertices arbitrarily, indexed by $k$. For every one of these $q$ sets of three vertices $s^k= \{ s^k_1, s^k_2, s^k_3 \}$, create 8 new vertices,  $\{ s^k_{i,j} : i \in [2], j \in [4] \}$, and 10 new sets $\{ T^k_{i,j} : i \in [2], j \in [5] \}$. The sets $T^k_{i,j}$ are defined as follows: $T^k_{i,1} = \{s^k_{i,1},s^k_{i,2}, s^k_{i,3} \}$, $T^k_{i,2} = \{ s^k_{i,2},s^k_{i,3}, s^k_{i,4} \}$,  $T^k_{i,3} = \{s^k_{i,1},s^k_{i,2}, s^k_{i,4} \}$, $T^k_{i,4} = \{s^k_{i,1},s^k_{i,3}, s^k_1 \}$, and $T^k_{i,5} = \{s^k_{i,4},s^k_2, s^k_3 \}$.
	
	The X33C instance is given by $(Y = X \cup \{ s^k_{i,j} : i \in [2], j \in [4], k \in [q] \}, \mathcal{T}= \mathcal{S} \cup \{ T^k_{i,j} : i \in [2], j \in [5], k \in [q] )$. Note that $|Y|=11q$. We show that $(Y,\mathcal{T})$ has an exact triple-cover if and only if there exists an exact cover for $(X, \mathcal{S})$.
	
	First, suppose that there exists an exact cover for $(X, \mathcal{S})$. Then there exists an exact triple-cover for $(Y, \mathcal{T})$ by selecting sets $T^k_{i,j}$ for every $k \in [q]$, $i \in [2]$, and $j \in [5]$. It is easy to check that these $10q$ sets cover each $s^k_{i,j}$ exactly three times, as well as covering $s^k_k$ exactly twice, for all $k \in [q]$ and $k \in [3]$. Hence, all we need to do is add the solution to the original X3C instance.
	
	Now, suppose that there exists an exact triple-cover by 3-sets for the X33C instance. This implies that, for any $k$ and $i$, exactly three sets from $\{ T^k_{i,1}, T^k_{i,2}, T^k_{i,3} \}$ must be chosen (recall that we can choose the same set more than once), because these are the only sets that contain $s^k_{i,2}$, which must be covered exactly three times. 
	We now consider how we can choose these three sets. Suppose that $T^k_{i,2}$ is chosen more than once. Then only (at most) one of $T^k_{i,1}$ and $T^k_{i,3}$ is chosen, so we still need to cover $s^k_{i,1}$ (at least) twice. The only way to do this is by choosing $T^k_{i,4}$ twice. But then $s^k_{i,3}$, which is contained in both $T^k_{i,2}$ and $T^k_{i,4}$, is covered more than three times, a violation of the conditions of an exact triple-cover. By similar reasoning, we can show that $T^k_{i,3}$ cannot be chosen more than once. 
	Now suppose that $T^k_{i,1}$ is chosen twice. Then it remains to choose exactly one of $T^k_{i,2}$ and $T^k_{i,3}$; suppose WLOG that we choose $T^k_{i,2}$. Then we still need to cover $s^k_{i,1}$ an additional time. The only way to do this is to choose $T^k_{i,4}$, which also covers $s^k_{i,3}$, meaning that $s^k_{i,3}$ is covered four times, violating the condition of the exact triple-cover.
	Finally, suppose that $T^k_{i,1}$ is chosen three times. Then we still need to cover $s^k_{i,4}$ three times without covering any of $s^k_{i,1}$, $s^k_{i,2}$, or $s^k_{i,3}$ again. We therefore need to choose $T^k_{i,5}$ three times.
	Otherwise, we can choose each of $T^k_{i,1}$, $T^k_{i,2}$, $T^k_{i,3}$, $T^k_{i,4}$, and $T^k_{i,5}$ once each, which covers each of $s^k_{i,1}$, $s^k_{i,2}$, and $s^k_{i,3}$ once each. All other options have been ruled out. In particular, it is necessary to choose $T^k_{i,5}$ at least once.
	
	So there are two options. Regardless of which option we choose, we still have to cover each of $s^k_{i',1}$, $s^k_{i',2}$, and $s^k_{i',3}$ three times each, for $i' \not= i$. Since the options for $i'$ are symmetric to those for $i$, it is again necessary to choose $T^k_{i',5}$ at least once. If we choose $T^k_{i,1}$ three times and $T^k_{i,5}$ three times, as well as $T^k_{i',5}$ at least once (as we must), then $s^k_2$ and $s^k_3$ are covered at least four times, a violation of the exact triple-cover. Therefore the only possibility is to choose each of $T^k_{i,1}$, $T^k_{i,2}$, $T^k_{i,3}$, $T^k_{i,4}$, and $T^k_{i,5}$ once. Similarly, we must choose each of $T^k_{i',1}$, $T^k_{i',2}$, $T^k_{i',3}$, $T^k_{i',4}$, and $T^k_{i',5}$ once each also. And, since $k$ was arbitrary, this holds for all $k \in [q]$. 
	
	So, for all  $k \in [q]$ and all $i \in [2]$, each of $T^k_{i,1}$, $T^k_{i,2}$, $T^k_{i,3}$, $T^k_{i,4}$, and $T^k_{i,5}$ is chosen exactly once, a total of $10q$ sets chosen. We therefore have $q$ more sets to choose, which necessarily cover each of $v \in \mathcal{S}$ exactly once (note that each $v \in \mathcal{S}$ corresponds to an $s^k_j$ for some $k \in [q]$ and $j \in [3]$, and these are covered exactly once by either $T^k_{1,4}$ or $T^k_{1,5}$, and exactly once again by either $T^k_{2,4}$ or $T^k_{2,5}$). The only way to choose $q$ sets that cover each $v \in \mathcal{S}$ exactly once is by choosing an exact cover for the instance $(X, \mathcal{S})$.
	\end{proof}

Using this lemma, we can now show that finding an outcome satisfying \pps and PO is \nphard through a reduction from X33C. 

\begin{proof}[Proof of Theorem~\ref{thm:pps-po-np-hard}]
	Let $(Y,\mathcal{T})$ be an instance of X33C, with $|Y|=r$. Let $\epsilon \in (0,1/(3r))$. We define a \seqdecision problem as follows. There are $r$ players, one corresponding to each vertex $v \in Y$, and $r$ issues. For each issue, there are $|Y| + |\mathcal{T}|$ alternatives. For each issue $t$ and each player $i$, there is an alternative $a_{t,i}$ which is valued at $1-\epsilon$ by player $i$, and 0 by all other players. The remaining $|\mathcal{T}|$ alternatives correspond to the 3-sets from the X33C instance. For a set $T_j \in \mathcal{T}$, the corresponding alternative is valued at $\frac{1}{3}$ by players $i \in T_j$, and valued at 0 by all other players. Note that $\pps_i = 1-\epsilon$ for each player $i$, because there are exactly as many issues as players, and each player values its most preferred alternative for each issue at $1-\epsilon$.
	
	We now show that there exists an exact triple-cover by 3-sets if and only if all outcomes to the \seqdecision problem that satisfy \pps and PO have $u_i(\vc)=1$ for all $i$.
	First, suppose that there exists an exact triple-cover by 3-sets. We need to show that all outcomes satisfying \pps and PO have $u_i(\vc)=1$ for all $i$. So suppose otherwise -- that there exists an outcome satisfying \pps and PO with $u_i(\vc) \not= 1$ for some player $i$. In particular, some player must have $u_i(\vc) > 1$, otherwise $\vc$ is not PO (because it is possible to choose an outcome corresponding exactly to an exact triple-cover, which gives each player utility 1). But players only derive utility in discrete amounts of $1-\epsilon$ or $\frac{1}{3}$, which means that any player with $u_i(\vc)>1$ has $u_i(\vc) \ge \frac{4}{3}-\epsilon$. 
	\[ \sum_{i=1}^r u_i(\vc) \ge \frac{4}{3}-\epsilon + \sum_{i=1}^{r-1} (1-\epsilon) = \frac{1}{3} + r - r\epsilon >r, \]
	where the last inequality holds because $\epsilon < 1/(3r)$. However, this is a contradiction because each alternative in each of the $r$ issues contributes at most $1$ to the social welfare. Therefore, every outcome satisfying \pps and PO has $u_i(\vc)=1$ for all $i$.
	
	Next, suppose that there does not exist an exact triple-cover by 3-sets. So if we choose an alternative corresponding to a 3-set for every issue, it is not possible for every player to derive utility 1. Therefore, some player must derive utility $\frac{2}{3}$ (or lower), which violates \pps. Thus, every outcome that satisfies \pps must include at least one issue where the chosen alternative is one that corresponds to a player, not to a 3-set. Such an alternative only contributes $1-\epsilon$ to social welfare. Therefore, the social welfare is strictly less than $r$, which means that some player gets utility strictly less than 1. Therefore, there is no outcome satisfying \pps (either with or without PO) such that $u_i(\vc)=1$ for all $i$. Since the set of outcomes satisfying \pps is always non-empty, it is therefore not the case that all outcomes satisfying \pps and PO have $u_i(\vc)=1$ for all $i$.
\end{proof}

Because every outcome satisfying \rrs also satisfies \pps, we have the following corollary.
\begin{corollary} \label{cor:rrs-po-np-hard}
It is \nphard to find an outcome satisfying \rrs and PO. 
\end{corollary}


For \pgd, we show, in stark contrast to Theorem~\ref{thm:pps-po-np-hard}, that we can find an allocation satisfying \pps and PO in polynomial time. This is achieved using Algorithm~\ref{alg:pps-po}. Interestingly, it produces not an arbitrary allocation satisfying \pps and PO, but an allocation that assigns at least $p = \floor{m/n}$ goods to every player --- implying \pps, and maximizes weighted (utilitarian) social welfare according to some weight vector --- implying PO. 

\medskip
\begin{algorithm}[ht]
	\SetAlgoLined
	\KwIn{The set of players $N$, the set of private goods $M$, and players' utility functions $\{u_i\}_{i \in N}$}
	\KwOut{A deterministic allocation $\vA$ satisfying \pps and PO}
	$\vec{w} \gets (1/n,\ldots,1/n) \in \bbR^n$\;
	$\vA \gets \argmax_{\vec{A'}} \sum_{i \in N} w_i \cdot u_i(\vec{A'})$\;\label{alg:init-alloc}
	\While(\tcc*[f]{Until every player receives at least $p = \floor{m/n}$ goods}){$\exists i \in N, |A_i| < p$}
	{\label{alg:outer-start}
		$GT \gets \{i \in N : |A_i| > p\}$\tcc*{Partition players by the number of goods they receive} 
		$EQ \gets \{i \in N : |A_i| = p\}$\; 
		$LS \gets \{i \in N : |A_i| < p\}$\;
		$DEC = GT$\tcc*{Players whose weights we will decrease}\label{alg:dec-start}
		\While{$DEC \cap LS = \emptyset$}{
			\label{alg:inner1-start}
			\tcc{Minimally reduce weights of players in $DEC$ so a player in $DEC$ loses a good}
			$(i^*,j^*,g^*) \gets \argmin_{i \in DEC, j \in N\setminus DEC, g \in A_i} (w_i\cdot v_{i,g})/(w_j \cdot v_{j,g})$\;\label{alg:reduce-start}
			$r \gets (w_{i^*}\cdot v_{i^*,g^*})/(w_{j^*} \cdot v_{j^*,g^*})$\;\label{alg:r-select}
			$\forall i \in DEC,\,\, w_i \gets w_i / r$\;\label{alg:reduce-end}
			$DEC \gets DEC \cup \{j^*\}$\;\label{alg:add}
			$D(j^*) \gets (i^*,g^*)$\tcc*{Bookkeeping: $j^*$ can receive $g^*$ from $i^*$}
		}\label{alg:inner1-end}		
		$j^* \gets DEC \cap LS$\tcc*{Player from $LS$ who receives a good}
		\While{$j^* \notin GT$}{
			\label{alg:inner2-start}
			$(i^*,g^*) \gets D(j^*)$\;
			$A_{i^*} \gets A_{i^*} \setminus \{g^*\}$\;
			$A_{j^*} \gets A_{j^*} \cup \{g^*\}$\;
			$j^* \gets i^*$\;
		}\label{alg:inner2-end}
	}\label{alg:outer-end}
	\Return $\vA$\;
	\caption{Polynomial time algorithm to achieve \pps and PO for private goods}
	\label{alg:pps-po}
\end{algorithm}
\medskip

At a high level, the algorithm works as follows. It begins with an arbitrary weight vector $\vec{w}$, and an allocation $\vA$ maximizing the corresponding weighted (utilitarian) social welfare. Then, it executes a loop (Lines~\ref{alg:outer-start}-\ref{alg:outer-end}) while there exists a player receiving less than $p$ goods, and each iteration of the loop alters the allocation in a way that one of the players who received more than $p$ goods loses a good, one of the players who received less than $p$ goods gains a good, and every other player retains the same number of goods as before. 

Each iteration of the loop maintains a set $DEC$ of players whose weight it reduces. Initially, $DEC$ consists of players who have more than $p$ goods (Line~\ref{alg:dec-start}). When the weights are reduced enough so that a player in $DEC$ is about to lose a good to a player, necessarily outside $DEC$ (Lines~\ref{alg:reduce-start}-\ref{alg:reduce-end}), the latter player is added to $DEC$ (Line~\ref{alg:add}) before proceeding further. When a player who has less than $p$ goods is added to $DEC$, this process stops and the algorithm leverages the set of ties it created along the way to make the aforementioned alteration to the allocation (Lines~\ref{alg:inner2-start}-\ref{alg:inner2-end}). 

We now formally state that this produces an allocation satisfying \pps and PO, and that it runs in polynomial time; the proof is deferred to the appendix. 
%

\begin{theorem}
	For \pgd, \pps and PO can be satisfied in polynomial time.
	\label{THM:PPS-PO-POLYTIME}
\end{theorem}


The complexity of finding an allocation (of private goods) satisfying the stronger guarantee \rrs along with PO in polynomial time remains open, as does the complexity of finding an allocation satisfying \propone and PO.

We note that the convenient approach of weighted welfare maximization we use in Theorem~\ref{THM:PPS-PO-POLYTIME} cannot be used for finding an outcome satisfying \rrs and PO, as the following example shows. This leads us to conjecture that it may be \nphard to find such an outcome.

\begin{example}
	Consider a private goods division problem with two players and four goods. Player 1 has utilities $(4,4,1,1)$ and player 2 has utilities $(3,3,2,2)$ for goods $(g_1, g_2, g_3, g_4)$, respectively. Note that $\rrs_1=\rrs_2=5$. Consider assigning weights $w_1$ and $w_2$ to players $1$ and $2$, respectively. If $4w_1>3w_2$, i.e., $w_1 > 3w_2/4$ then player 1 receives both $g_1$ and $g_2$, which means that player 2 receives utility less than her \rrs share. On the other hand, if $3w_2>4w_1$, i.e., $w_1<3w_2/4$ then player 2 receives both $g_1$ and $g_2$, which means that player 1 receives utility less than her \rrs share. 
	
	The only remaining possibility is that $w_1=3w_2/4$, but in that case, player 2 receives both $g_3$ and $g_4$. Regardless of how we divide goods $g_1$ and $g_2$, one of the two players still receives utility less than her \rrs share.
\end{example}

In contrast, a simple modification of Algorithm~\ref{alg:pps-po} seems to quickly find an allocation satisfying \propone and PO in hundreds of thousands of randomized simulations. 
At each iteration of this version, the set $DEC$ initially consists of players who attain their proportional share (it is easy to show using the Pigeonhole principle that this set is non-empty for any weighted welfare maximizing allocation), and ends when a player is added to $DEC$ that is not currently achieving \propone. Thus, at every loop, a player that was receiving her proportional share may lose a good (but will still achieve at least \propone), the player added to $DEC$ that was not achieving \propone gains a good, and some players that were achieving \propone but not their proportional share may lose a good, but only if they gain one too. These three classes of players are therefore analogous to players with more than $p$ goods, less than $p$ goods, and exactly $p$ goods in Algorithm~\ref{alg:pps-po}. Unfortunately, we are unable to prove termination of this algorithm because it is possible that a player who achieves \propone but not her proportional share loses a high-valued good while gaining a low-valued good, thus potentially sacrificing \propone. Thus we do not get a property parallel to the key property of Algorithm~\ref{alg:pps-po}, that no player's utility ever drops below her \pps share, after she attains it. However, our algorithm always seems to terminate quickly and finds an allocation satisfying \propone and PO in our randomized simulations, which leads us to conjecture that it may be possible to find an allocation satisfying \propone and PO in polynomial time, either from our algorithm directly or via some other utilization of weighted welfare maximization.

%% file: 5-disc.tex
\section{Discussion}
\label{sec:disc}

We introduced several novel fairness notions for \seqdecision and considered their relationships to existing mechanisms and fairness notions. Throughout the paper, we highlighted various open questions including the existence (and complexity) of a mechanism satisfying \rrs, \propone, and PO, the complexity of finding an outcome satisfying \propone and PO (for public decisions and private goods), the complexity of finding an outcome satisfying \rrs and PO (for private goods), and whether MNW provides a constant approximation to \rrs better than $1/2$. 

So far we only considered deterministic outcomes. If randomized outcomes are allowed (an alternative interpretation in the private goods case is that the goods are infinitely divisible), then the MNW solution satisfies \prop as a direct consequence of it satisfying \propone for deterministic outcomes (Theorem~\ref{thm:mnw-prop1}).\footnote{Of course, the \emph{realization} may fail to satisfy \prop (and other desiderata), but the \emph{lottery} is fair if we consider expected utilities.} To see this, consider replicating each \issue $K$ times and dividing utilities by $K$. The relative effect of granting a single player control of a single \issue becomes negligible. Thus, as $K$ approaches infinity, the utility of each player $i$ in an MNW outcome approaches a value that is at least their proportional share $\prop_i$. The fraction of copies of \issue $t$ in which outcome $a^t_j$ is selected can be interpreted as the weight placed on $a_j^t$ in the randomized outcome. Because \rrs, \pps, and \propone are relaxations of \prop, the randomized MNW outcome also satisfies all of them. 

For \pgd, this can be seen as a corollary of the fact that the randomized MNW outcome satisfies envy-freeness, which is strictly stronger than proportionality. This hints at a very interesting question: Is there a stronger fairness notion than proportionality in the \seqdecision framework that generalizes envy-freeness in \pgd? Although such a notion would not be satisfiable by deterministic mechanisms, it may be satisfied by randomized mechanisms, or it could have novel relaxations that may be of independent interest.

At a high level, our work provides a framework bringing together two long-studied branches of social choice theory --- fair division theory and voting theory. Both have at their heart the aim to aggregate individual preferences into a collective outcome that is fair and agreeable, but approach the problem in different ways. Fair division theory typically deals with multiple private goods, assumes access to cardinal utilities, and focuses on notions of fairness such as envy-freeness and proportionality. Voting theory, in contrast, typically deals with a single public decision (with the exception of combinatorial voting mentioned earlier), assumes access only to less expressive ordinal preferences, has the ``one voter, one vote'' fairness built inherently into the voting rules, and focuses on different axiomatic desiderata such as Condorcet consistency and monotonicity. 

Of course, one can use a voting approach to fair division, since we can have players express preferences over complete outcomes, and this approach has been used successfully to import mechanisms from voting to fair division and vice versa~\cite{ABS13,AS14}. However, not only does this approach result in an exponential blowup in the number of alternatives, it also does not provide a convenient way to express fair-division-like axioms. 
Continuing to explore connections between the two fields remains a compelling direction for future work.

%% file: 6-appendix.tex
\section{Relationships Among Fairness Axioms}
\label{app:fairness-axioms}

In this section, we analyze the relationship between the fairness properties we introduce in this paper, namely \rrs, \pps, and \propone. First, it is easy to show that \propone does not give any approximation to \rrs or \pps, both for public decisions and for private goods, because it is easy to construct examples where a player receives zero utility, still satisfies \propone, but has non-zero \pps share. 

In the other direction, for public decisions, we showed that \rrs implies $1/2$-\propone (Lemma~\ref{lem:rrs-propone}). For private goods, we can refine this result a bit further.

\begin{theorem} \label{thm:rrs-prop1}
For \pgd, \rrs implies \propone if and only if $m \le 4n-2$.
\end{theorem}

\begin{proof}
First, let us assume $m > 4n-2$. Consider the following fair division instance. Player $1$'s values, in the descending order, are as follows. 
\[
u_{max}^j(1) = \left\{
\begin{array}{ll}
(k-1) \cdot n + 1 &\mbox{if } j=1,\\
n &\mbox{if } j \in \{2,\ldots,k\cdot n-1\},\\
1 &\mbox{if } j \ge k \cdot n.
\end{array}
\right.
\]
It is easy to check that player $1$'s RRS share is exactly $u_{max}^1(1)$. Consider an allocation in which player $1$ receives only his most valuable good, and the remaining goods are partitioned among the other players arbitrarily. For the sake of completeness, let each other player have value $1$ for each good he receives under this allocation, and $0$ for the remaining goods. Hence, the allocation satisfies RRS.

Now, player $1$'s proportional share is given by
\begin{align*}
\frac{(k-1)\cdot n + 1 + (k n-2)\cdot n + (m-k n+1) \cdot 1}{n}
&= \frac{m + k n^2 - 3 n + 2}{n}\\
&> \frac{n + k n^2}{n}\\
&=k n + 1,
\end{align*}
where the second transition follows because $m > 4n-2$. 

The highest value that player $1$ can achieve by adding one more good to his allocation is $(k-1)\cdot n + 1 + n = k n + 1$, which falls short of the proportional share. Hence, the allocation is not Prop1. 

Now, let us assume that $m \le 4n-2$. Hence, $k \le 3$. Take a fair division instance, and let us focus on a player $i$. For the sake of notational convenience, we define $u_{max}^j(i) = 0$ for $j \in \{m+1,\ldots,4n-2\}$. Note that this affects neither his RRS share nor his satisfaction of Prop1. 

We now show that if player $i$ receives at least as much value as his RRS share $u_{max}^n(1) + u_{max}^{2n}(1) + u_{max}^{3n}(1)$, then we can make player $i$ receive his proportional share by adding a single good to his allocation. 

If player $i$ does not receive his most valuable good, then this can be accomplished by adding his most valuable good to his allocation because
\begin{align*}
&u_{max}^1(i)+u_{max}^n(1) + u_{max}^{2n}(1) + u_{max}^{3n}(1) \\
&\quad\ge \frac{\sum_{j=1}^n u_{max}^j(i)}{n} + \frac{\sum_{j=n+1}^{2n} u_{max}^j(i)}{n} + \frac{\sum_{j=2n+1}^{3n} u_{max}^j(i)}{n} + \frac{\sum_{j=3n+1}^{4n-2} u_{max}^j(i)}{n} \\
&\ge \frac{\sum_{j=1}^m u_{max}^j(i)}{n}.
\end{align*}
The first transition follows because $u_{max}^j(i) \ge v_{max}^{j+1}(i)$ for all $j \in [4n-3]$. 

Suppose player $i$ receives his most valuable good. Let $t$ be the smallest index such that player $i$ does not receive his $t^{\text{th}}$ most valuable good. Hence, $t \ge 2$. Let $u_i$ denote the utility to player $i$ under the current allocation. Then, we have that 
\begin{align}
u_i &\ge \sum_{j=1}^{t-1} u_{max}^j(i). \label{eqn:rrs-prop1-1}\\
u_i &\ge v_{max}^n(i) + u_{max}^{2n}(i) + v_{max}^{3n}(i). \label{eqn:rrs-prop1-2}
\end{align}
Multiplying Equation~\eqref{eqn:rrs-prop1-1} by $1/n$ and Equation~\eqref{eqn:rrs-prop1-2} by $(n-1)/n$, and adding, we get
\begin{align*}
u_i &\ge \frac{\sum_{j=1}^{t-1} u_{max}^j(i)}{n} + \frac{n-1}{n} \cdot (u_{max}^n(i) + u_{max}^{2n}(i) + u_{max}^{3n}(i))\\
&\ge \frac{\sum_{j=1}^{t-1} u_{max}^j(i) + \sum_{j=n+1}^{2n-1} u_{max}^j(i) + \sum_{j=2n+1}^{3n-1} u_{max}^j(i) + \sum_{j=3n}^{4n-2} u_{max}^j(i)}{n}\\
&= \frac{\sum_{j=1}^{t-1} u_{max}^j(i) + \sum_{j=n+1}^{4n-2} u_{max}^j(i) - u_{max}^{2n}(i)}{n}\\
&\ge \frac{\sum_{j=1}^{t-1} u_{max}^j(i) + \sum_{j=n+2}^{4n-2} u_{max}^j(i)}{n}.
\end{align*}

If $t \ge n+2$, player $i$ already receives his proportional share. Otherwise, let us now add player $i$'s $t^{\text{th}}$ most valuable good to his allocation. His utility increases to 
\[
u_i + u_{max}^t(i) \ge \frac{\sum_{j=1}^{t-1} u_{max}^j(i) + \sum_{j=n+2}^{4n-2} u_{max}^j(i)}{n} + \frac{\sum_{j=t}^{n+1} u_{max}^j(i)}{n} = \frac{\sum_{j=1}^m u_{max}^j(i)}{n},
\]
where the first transition follows because $t \ge 2$. Thus, player $i$ receives his proportional share after adding a single good to his allocation. Because player $i$ was chosen arbitrarily, we have that the allocation satisfies \propone. 
\end{proof}

\section{Proof of Lemma~\ref{LEM:MAX-REDUCTION}}

We will say that a set $\set{x_1,\ldots,x_n}$ of $n$ non-negative real numbers is \emph{feasible} for a given $0 < \delta <1$ if $\sum_{k=1}^n \max \{ 0, 1- x_k\} \le \delta$.

Let $X=\set{x_1,\ldots,x_n}$ be a feasible set and let $i = \argmin_{k \in [n]} \set{x_k}$, so that $x_i$ is the minimum value in $X$. Suppose that $x_i > 1-\delta$. We will show that there necessarily exists another feasible set $X'=\set{x'_1,\ldots,x'_n}$ with $\prod_{i=1}^n x'_i < \prod_{i=1}^n x_i$.

To that end, let $j = \argmin_{k \in [n] \setminus \{ i \}} \set {x_k}$, so that $x_j$ is the second-lowest value in $X$. If $x_j \ge 1$ then set $x'_i=1-\delta$ and $x'_k=x_k$ for all $k \in [n] \setminus \{ i \}$. Clearly this set $X'$ is feasible since 
\[ \sum_{k=1}^n \max \{ 0, 1- x_k\} = 1-x_i=\delta, \]
and 
\[ \prod_{k=1}^n x'_k = x'_i \prod_{k \not= i} x'_k < x_i \prod_{k \not= i} x_k = \prod_{k=1}^n x_k. \]

Now consider the case where $x_j < 1$. Let $\epsilon = \frac{1}{2} (1-x_j)$ and define $x_i'=x_i-\epsilon$, $x_j'=x_j+\epsilon$, and $x'_k=x_k$ for all $k \in [n] \setminus \{ i ,j \}$.. For feasibility of $X'$,
\begin{align*} \sum_{k=1}^n \max \{ 0, 1- x'_k\} &= \sum_{k \in [n] \setminus \{ i,j \}} \max \{ 0, 1-x'_k \} + \max \{ 0, 1-x_i+\epsilon \} + \max \{ 0, 1-x_i+\epsilon \} \\
	&= \sum_{k \in [n] \setminus \{ i,j \}} \max \{ 0, 1-x_k \} + (1-x_i+\epsilon)+(1-x_j-\epsilon)\\
	&= \sum_{k \in [n] \setminus \{ i,j \}} \max \{ 0, 1-x_k \} + (1-x_i)+(1-x_j)\\
	&= \sum_{k=1}^n \max \{ 0, 1- x_k\},
	\end{align*}
and 
\[ \prod_{k=1}^n x'_k = x'_i x'_j \prod_{k \not= i,j} x'_k = (x_i-\epsilon)(x_j+\epsilon) \prod_{k \not= i,j} x_k = (x_ix_j - (x_j-x_i)\epsilon-\epsilon^2) < x_ix_j \prod_{k \not= i,j} x_k = \prod_{k=1}^n x_k. \]
To complete the proof, it remains to show that any feasible set with smallest entry $x_i=1-\delta$ has $\prod_{k=1}^n x_k \ge 1-\delta$. Note that if $x_i=1-\delta$ then for $X$ to be feasible it can not be the case that $x_j<1$. Therefore,
\[ \prod_{k=1}^n x_k = (1-\delta) \prod_{k \not= i} x_k \ge (1-\delta). \]
\qed

\section{Proof of Theorem~\ref{THM:PPS-PO-POLYTIME}}

	First, we prove that Algorithm~\ref{alg:pps-po} produces an allocation satisfying \pps and PO. In particular, we prove that at the end, allocation $\vA$ satisfies (A) $|A_i| \ge p$ for every player $i \in N$, and (B) $\vA$ maximizes the weighted social welfare $\sum_{i \in N} w_i \cdot u_i(A_i)$. Note that property (A) implies that $\vA$ satisfies \pps, and property (B) implies that $\vA$ is PO. 
	
	Before we prove these claims, we note that $\vA$ maximizes the weighted social welfare if and only if it allocates each good $g \in M$ to a player $i$ maximizing $w_i \cdot u_i(g)$. 
	
	First, we prove claim (A). Because the outer loop (Lines~\ref{alg:outer-start}-\ref{alg:outer-end}) explicitly runs until each player receives at least $p$ goods, we simply need to prove that the loop terminates. For this, we need to analyze the first inner loop (Lines~\ref{alg:inner1-start}-\ref{alg:inner1-end}) and the second inner loop (Lines~\ref{alg:inner2-start}-\ref{alg:inner2-end}). 
	
	Each iteration of the first inner loop reduces the weights of players in $DEC$, and adds a player $j^* \in N\setminus DEC$ to $DEC$. Because $|DEC|$ increases by one in each iteration, the first inner loop terminates after $O(n)$ iterations. 
	
	The second inner loop starts from the player $j^* \in LS$ that was added to $DEC$ at the end of the first inner loop, and traces back to the player $i^*$ that was about to lose good $g^*$ to $j^*$ when $j^*$ was added to $DEC$. The good is explicitly transferred, and if player $i^*$ had exactly $p$ goods initially, the algorithm continues to find another good to give back to player $i^*$ by tracing back to the conditions under which player $i^*$ was added to $DEC$. This way, the transfers add a good to a player who was in $LS$, maintain the number of goods of the players who were in $EQ$, and remove a good from a player who was in $GT$. Because in each iteration, player $i^*$ had to be present in $DEC$ before player $j^*$ was added, this loop cannot continue indefinitely, and must terminate in $O(n)$ iterations as well. 
	
	Thus, both inner loops terminate in $O(n)$ iterations, and in each iteration of the outer loop, a player in $LS$ receives an additional good without any new players being added to $LS$. This monotonically reduces the metric $\sum_{i \in LS} p-|A_i|$ by at least $1$ in each iteration. Because this metric can be at most $p \cdot n \le m$ to begin with, the outer loop executes $O(m)$ times. 
	
	We now prove claim (B). That is, we want to show that allocation $\vA$ remains a maximizer of the weighted social welfare according to weight vector $\vec{w}$ at during the execution of the algorithm. Because we specifically select $\vA$ to be a weighted social welfare maximizer in Line~\ref{alg:init-alloc}, we simply show that neither the weight update in the first inner loop nor the changes to the allocation in the second inner loop violate this property. 
	
	In the first inner loop, because the weights of players in $DEC$ are reduced by the same multiplicative factor, goods can only transfer from players in $DEC$ to players in $N\setminus DEC$. However, the choice of $r$ in Line~\ref{alg:r-select} ensures that the weight reduction stops when the first such potential transfer creates a tie, preserving allocation $\vA$ as a weighted welfare maximizer. Alterations to allocation $\vA$ during the second inner loop also do not violate this property because this loop only transfers a good $g^*$ from player $i^*$ to player $j^*$ when the two players were anyway tied to receive the good. 
	
	This concludes our claim that the algorithm terminates, and correctly produces an allocation satisfying \pps and PO. We already established that the outer loop executes $O(m)$ times, and the two inner loops execute $O(n)$ times. The bottleneck within the inner loops is the $\argmin$ computation in Line~\ref{alg:reduce-start}, which requires $O(n \cdot m)$ time to find the minimum across all goods owned by players in $DEC$ and all players outside $DEC$. Consequently, the asymptotic running time complexity of the algorithm is $O(m \cdot n \cdot n m) = O(n^2 \cdot m^2)$.
\qed